\documentclass[aps,prl,twocolumn,superscriptaddress,groupedaddress,longbibliography]{revtex4}

\usepackage{graphicx}
\usepackage{dcolumn}
\usepackage{bm}
\usepackage{amssymb}   
\usepackage{physics}
\usepackage[dvipsnames]{xcolor}
\usepackage{amsthm}   
\usepackage{amssymb}   
\usepackage{physics}
\usepackage{times}
\usepackage{newtxmath}
\usepackage[colorlinks = true,
            linkcolor = orange,
            urlcolor  = orange,
            citecolor = orange,
            anchorcolor = orange, 
            allcolors = red]{hyperref}

\newtheorem{theorem}{Result}
\newtheorem{lemma}{Lemma}
\newtheorem{task}{Task}
\newtheorem{corollary}{Corollary}

\newcommand{\choi}[1]{\mathcal{J}^{\mathcal{#1}}}

\newcommand{\maxKB}[0]{\ketbra{\Phi^{+}}}

\newcommand{\suppVector}[1]{\abs{ \textrm{supp}(#1) } } 
\newcommand{\choiRank}[0]{r_{c}^{\mathcal{N}}}

\newcommand{\depola}[1]{\mathcal{D}^{\rm{pol}}_{#1}} 
\newcommand{\dephase}[1]{\mathcal{D}^{\rm{ph}}_{#1}} 

\newcommand{\ben}[1]{{\color{black}#1}} 
\newcommand{\CY}[1]{{\color{black}#1}}
\newcommand{\CYtwo}[1]{{\color{black}#1}}
\newcommand{\bentwo}[1]{{\color{black}#1}}

\begin{document}

\preprint{APS/123-QED}

\title{Operational Interpretation of the Choi Rank Through Exclusion Tasks}

\author{Benjamin Stratton}
\email{ben.stratton@bristol.ac.uk}
\affiliation{Quantum Engineering Centre for Doctoral Training, H. H. Wills Physics Laboratory and Department of Electrical \& Electronic Engineering, University of Bristol, BS8 1FD, UK}
\affiliation{H.H. Wills Physics Laboratory, University of Bristol, Tyndall Avenue, Bristol, BS8 1TL, UK}

\author{Chung-Yun Hsieh}
\affiliation{H.H. Wills Physics Laboratory, University of Bristol,
Tyndall Avenue, Bristol, BS8 1TL, UK}

\author{Paul Skrzypczyk}
\affiliation{H.H. Wills Physics Laboratory, University of Bristol,
Tyndall Avenue, Bristol, BS8 1TL, UK}
\affiliation{CIFAR Azrieli Global Scholars Program, CIFAR, Toronto Canada}

\date{\today}

\begin{abstract}
The Choi-state is an indispensable tool in the study and analysis of quantum channels. Considering a channel in terms of its associated Choi-state can greatly simplify problems. 
It also offers an alternative approach to the characterisation of a channel, with properties of the Choi-state providing novel insight into a channel's behaviour. The rank of a Choi-state, termed the {\em Choi-rank}, has proven to be an important characterising property, and here, its significance is further elucidated through an operational interpretation. The Choi-rank is shown to provide a universal bound on how successfully two agents, Alice and Bob, can perform an entanglement-assisted exclusion task. The task can be considered an extension of super-dense coding, where Bob can only output information about Alice's encoded bit-string with certainty. Conclusive state exclusion, in place of state discrimination, is therefore considered at the culmination of the super-dense coding protocol. In order to prove this result, a necessary condition for conclusive $k$-state exclusion of a set of states is presented in order to achieve this result, and the notions of weak and strong exclusion are introduced.
\end{abstract}

\maketitle

{\bf\em Introduction.---}
States give us only a snapshot in time. To model how systems evolve, interact with other systems, and respond to external stimuli, it is essential to understand dynamics. Such understanding then enables the prediction of a system's future state, facilitates the design and implementation of controls for on-demand manipulation, and allows for the characterisation of a system's response to external influences, such as noise.

In closed quantum systems, dynamics is modelled by unitary operators; the evolution of a state $\rho$ to $\rho'$ is given by $\rho' = U \rho U^{\dagger}$ for some unitary $U$. A more general notion of quantum dynamics is captured by {\em quantum channels}, or simply {\em channels}, which are {\em completely-positive trace-preserving} (CPTP) linear maps~\cite{nielsen_chuang_2010,wolf2012quantum}. Operationally, channels can be thought of as modelling the dynamics of open quantum systems, where a system evolves whilst interacting with an environment. Any channel $\mathcal{N}$ acting on a system ${\rm S}$ has a {\em Stinespring dilation}~\cite{Stinespringdilation} given by
$
\mathcal{N}(\rho) = \textrm{tr}_{\textrm{E}} \big[ U_{\textrm{SE}} ( \rho_{\textrm{S}} \otimes \tau_{\textrm{E}}) U_{\textrm{SE}}^{\dagger} \big], 
$
where $\textrm{S},\textrm{E}$ represent the system and environment respectively, $\tau_{\textrm{E}}$ is some environment state, and $U_{\rm SE}$ is a unitary acting on ${\rm SE}$. All quantum dynamics can therefore be modelled as unitary with respect to some higher dimensional space, and can be described by a single state ($\tau_{\rm E}$) and unitary ($U_{\rm SE}$). Whilst the Stinespring dilation gives a physically motivated description of quantum dynamics, a more mathematically motivated description is given through a {\em Kraus decomposition}~\cite{nielsen_chuang_2010}\footnote{ \ben{
Note, other methods of modelling the dynamics of open quantum systems beyond quantum channels does exist, one could also consider microscopically derived master equations or the Von Neumann equation, for example.} 
}. For a channel $\mathcal{N}$ there always exists a set of $M$ operators $\{K_{x}\}^{M}_{x=1}$, where $\sum_{x=1}^{M} K_{x}^{\dagger}K_{x} = \mathbb{I}$, such that the action of $\mathcal{N}$ on a state $\rho$ is given by
$
\mathcal{N}(\rho) = \sum_{x=1}^{M} K_{x}\rho K_{x}^{\dagger}. 
$
The Kraus decomposition is useful when applying \CY{a channel} since one only needs to consider the input state and not the environment.

\ben{A given channel can have infinitely many Stinespring dilations and Kraus representations -- all sets of operators that are unitarily equivalent to $\{K_{x}\}^{M}_{x=1}$ define the same channel, for example.
However, and perhaps surprisingly,} a channel can be {\em uniquely} described through its action on a single quantum state. The {\em Choi-Jamio{\l}kowski isomorphism}~\cite{Jamiolkowski1972,ChoiOriginal} is a linear mapping between quantum channels and bipartite quantum states. For a channel $\mathcal{N}_{\rm A}$ acting on a system ${\rm A}$ with dimension $d$, its {\em Choi-state} $\choi{N}_{\rm AB}$ is a bipartite state in ${\rm AB}$ (where ${\rm B}$ also has dimension $d$) defined by
\begin{equation}
    \choi{N}_{\rm AB} = (\mathcal{N}_{\rm A} \otimes \mathcal{I}_{\rm B}) \big( \ketbra{\Phi^+}_{\rm AB} \big), 
\end{equation}
where $\ket{\Phi^+}_{\rm AB} \coloneqq\sum_{i=0}^{d-1}\ket{ii}_{\rm AB}/\sqrt{d}$ is a maximally entangled state in ${\rm AB}$~\footnote{In general $\ket{\Phi}_{\rm AB}$ can be any full-Schmidt-rank pure state.} ($\ket{i}_{\rm A}$ and $\ket{i}_{\rm B}$ are elements of fixed orthonormal bases of ${\rm A},{\rm B}$, respectively), and $\mathcal{I}_{\rm B}$ is an identity channel acting on ${\rm B}$. Subscripts will be used to explicitly denote the corresponding (sub-)systems if needed. From the Choi-state, the action of $\mathcal{N}_{\rm A}$ on a state $\rho$ can be recovered as 
$
\mathcal{N}_{\rm A}(\rho) = d ~ \textrm{tr}_{\rm B}\left[ \left(\mathbb{I}_{\rm A} \otimes \rho^{t}\right) \choi{N}_{\rm AB}\right],
$
where $(\cdot)^{t}$ is the transpose operation in the given fixed basis. Hence, remarkably, one can fully characterise the action of $\mathcal{N}$ on arbitrary states by a single action on $\ket{\Phi^+}_{\rm AB}$.

Choi-states have proved to be one of the most powerful tools for understanding and characterising channels, both analytically and numerically. For instance, a linear map $\mathcal{N}_{\rm A}$ acting on ${\rm A}$ is a channel if and only if its Choi-state is positive semi-definite, $\choi{N}_{\rm AB} \geq 0$, and maximally mixed in the B subspace, $\textrm{tr}_{\rm A}[\choi{N}_{\rm AB}] = \mathbb{I}_{\rm B}/d$~\cite{Jamiolkowski1972,ChoiOriginal}. By applying further operations to \CY{Choi-states, channels can be categorised into relevant subsets---a task often challenging when relying on Stinespring dilation or Kraus decomposition~\cite{rains2001semidefinite, Wang_2019, PhysRevResearch.2.023298, ji2021convertibility}. For example, a channel is entanglement-breaking~\cite{Horodecki2003RevMP} if its Choi-state is separable; it is steering-breaking if its Choi-state is steerable~\cite{Ku2022PRXQ}. Also, non-signalling conditions can be described by simple formulae of Choi-states~\cite{PianiPRA2006,Hsieh2022PRR,WangITIT2018,Hoban2018NJP,BeckmanPRA2001,Eggeling2002EPL,Duan2016ITIT}. Generally, Choi-states allow problems concerning quantum dynamics to be reformulated as numerically feasible convex optimisations~\cite{PaulSDPBook,TakagiPRX2019}, making them efficiently solvable. To date, Choi-states are widely recognised as an essential tool in characterising \ben{and} quantifying dynamical quantum signatures---they provide an alternative, often simpler, approach to dealing with channels~\cite{stratton2024, zanoni2024choidefined, gour2020dynamical, Gour_2021, PhysRevA.98.042338, PhysRevX.8.021033, doi:10.1098/rspa.2019.0251, Wang_2019, Berk_2021, PhysRevResearch.2.023298, bäuml2019resource, PhysRevLett.125.180505, ji2021convertibility,Haapasalo2021Quantum, Hsieh2024Quantum, liu2019resource, PhysRevResearch.2.012035, PhysRevLett.122.190405, PhysRevLett.115.010405, Hsieh_2020,Horodecki2003RevMP,Ku2022PRXQ,PaulSDPBook,Hsieh2022PRR,Hsieh2024,TakagiPRX2019,BuscemiITIT2010,WangITIT2018,Hoban2018NJP,BeckmanPRA2001,Eggeling2002EPL,Duan2016ITIT,PianiPRA2006}.} 

An important property of the Choi-state is its rank, $\choiRank$, termed {\em Choi-rank}. This is a key characterising property that provides insight into the structure of quantum dynamics; it has been shown to place mathematical bounds on the channel's description. For example, it is one, $\choiRank=1$, if and only if the channel is unitary \cite{Girard_2022}; it serves as a lower bound on the number of operators needed in the Kraus decomposition of the channel, $\choiRank \leq M$~\cite{nielsen_chuang_2010}; it equals the minimum dimension of the environment of a channel's Stinespring dilation~\cite{Singh2022}; and, when considering mixed unitary channels, is used to bound the number of unitaries needed to define the channel~\cite{Girard_2022}. 

To date, the Choi-rank has been used solely as a mathematical tool, lacking a clear operational interpretation. \ben{In this work, we provide the Choi-rank with an operational interpretation, further cementing its importance as a characterising property.} To this end, we introduce the task of \CY{{\em entanglement-assisted sub-channel exclusion}, \ben{phrasing them as} communication tasks that resemble} super-dense coding \cite{PhysRevLett.69.2881}, \ben{but} where state exclusion \cite{PhysRevA.66.062111, Bandyopadhyay_2014, Pusey2012} \CY{replaces} state discrimination \cite{Bae_2015}. 
\CY{We show that \ben{the} Choi-rank establishes a tight\ben{,} fundamental upper bound \ben{on one's ability to succeed in a task of this nature}. \bentwo{Hence, this bound can be used to provide the Choi-rank with a physical and operational meaning.}}

Informally, in the task of state exclusion, a referee gives a player a state from a set of $N$ possible predetermined states. The player then performs a measurement on the state and aims to exclude a set of $k$ states that were \textit{not} given to them. This leaves the player with $N-k$ possible states that they were sent. In comparison, if performing the task of state discrimination, after the measurement the player would aim to declare a single state that they \textit{were} sent. For some sets of states, the player can perform {\em conclusive} state exclusion --- excluding $k$ states with unit probability --- even when they can say nothing deterministically about what state they \textit{do have}. Exclusionary information~\cite{PhysRevLett.125.110401} --- knowledge about what state the player does not have --- can, therefore, be the only certain knowledge about the state that it is possible for the player to obtain. Hence, when using states to encode messages, one may be able to deterministically say what message was not encoded, whilst only being able to probabilistically say what message was encoded. The importance of exclusionary information has already been demonstrated in the foundations of quantum theory \cite{PhysRevA.66.062111, Pusey2012, hsieh2023quantum} and in the quantification of quantum resources~\cite{uola2020all, PhysRevLett.125.110401, Ye2021}. Here, we build on its ongoing significance in quantum information theory~\cite{Bandyopadhyay_2014, Heinosaari_2018, mishra2023optimal, johnston2023tight, Knee_2017, PhysRevA.107.L030202, PhysRevResearch.2.013326, Molina_2019} by establishing a connection between exclusionary information and the Choi-state.

{\bf\em State exclusion tasks.---}
Formally, in a state exclusion task~\cite{PhysRevA.66.062111, Pusey2012, Bandyopadhyay_2014}, a referee has a set of states $\{\rho_{x} \}^{N}_{x=1}$ and sends one state from the set, with probability $p_{x}$, to a player. The player performs a general $N$-outcome measurement described by a \textit{positive operator-valued measure} (POVM)~\cite{nielsen_chuang_2010} $\{T_{a}\}_{a=1}^{N}$, where $T_a\ge0$ $\forall\,a$ and $\sum_{a=1}^NT_a=\mathbb{I}$, on the state and outputs a label $g \in \{1, \ldots , N\}$. They win if $g \neq x$ and fail if $g=x$. Namely, the player wins if they successfully exclude the state by outputting a label that was not associated to the sent state; they fail if they output the label associated to the sent state. 

If the player outputs a single label $g$ such that $g \neq x$ with certainty, this is {\em conclusive $1$-state exclusion}. This occurs if the player is able to find a POVM such that  
\begin{equation}
    \textrm{tr} \big[ T_{x}\rho_{x}] = 0 ~~ \forall ~~ x  \in \{1, \ldots ,N\}. \label{stateExclusionEquation}
\end{equation}
If the player gets the measurement outcome associated to $T_{g}$, they output $g$ knowing with certainty the referee could not have sent $\rho_{g}$. 

If the player outputs a set of $k$ labels, $\{ g_{i} \}^{k}_{i=1}$, such that $x \notin \{ g_{i} \}^{k}_{i=1}$ with certainty, this is conclusive {\em $k$-state exclusion}. There are $N \choose k$ different sets of $k$ labels the player could exclude, corresponding to all the different subsets of $ \{1,\ldots,N \} $ of length $k$. Therefore, when performing $k$-state exclusion, the player aims to find a POVM with ${N \choose k}$ elements such that each measurement outcome allows the player to exclude a subset of states from $ \{ \rho_{x} \}^{N}_{x=1}$ of length $k$.

\begin{figure}
    \centering
     \hspace*{-1cm} \includegraphics[scale=0.38]{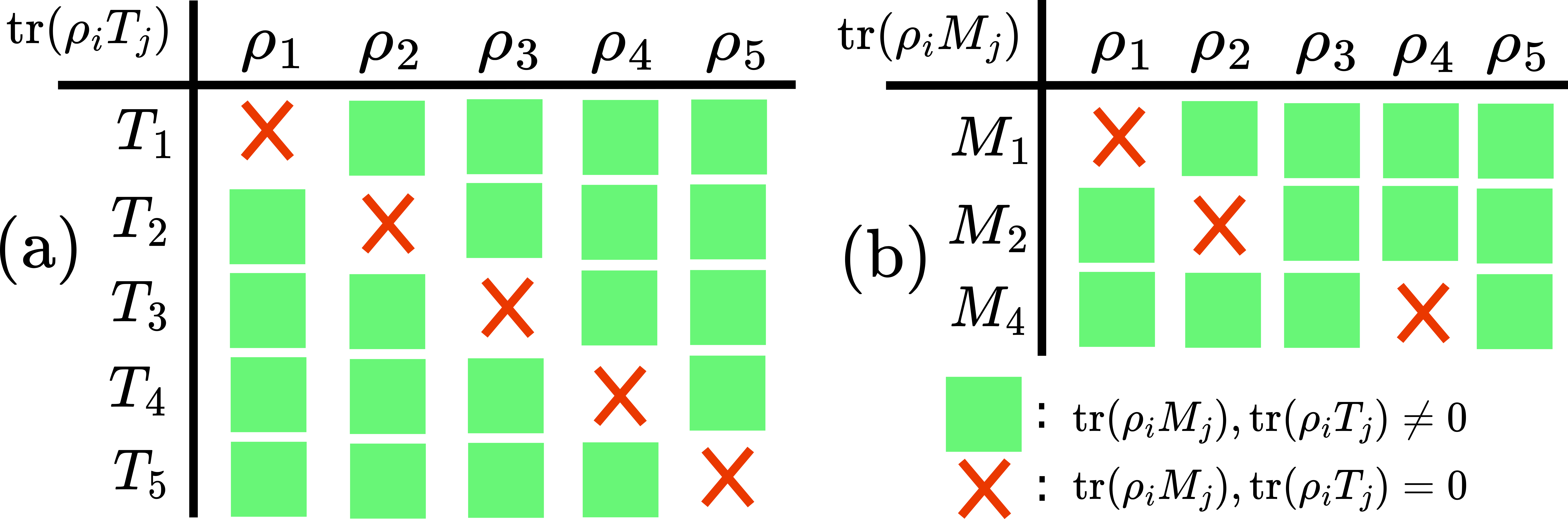}
    \caption{A graphical depiction of weak versus strong state exclusion. (a) Strong state exclusion: each outcome of the POVM $\{T_{a}\}^{5}_{a=1}$ excludes a different state from the set $\{ \rho_{x}\}^{5}_{x=1}$. (b) Weak state exclusion: each outcome of the POVM $\{M_{a}\}_{a\in\{1,2,4\}}$ excludes one state from the set $\{\rho_{x}\}^{5}_{x=1}$. However, only some states in $\{ \rho_{x}\}^{5}_{x=1}$ are ever excluded. There exists no POVM element in $\{M_{a}\}_{a\in\{1,2,4\}}$ to exclude $\rho_3$ or $\rho_5$.} 
    \label{weak_strong_exlcusion_figure}
\end{figure}

Whilst the notion of state exclusion is widely understood, some of the nuances in the definition are not agreed upon. In addition to Eq.~(\ref{stateExclusionEquation}) being a condition for conclusive state exclusion, the following additional condition, 
\begin{equation}
    \sum_{x=1}^{N} \textrm{tr} \big[ T_{a} \rho_{x} ] \neq 0~\forall~a\in\{1,...,N\},
    \label{addationalStateExclusionPOVM}
 \end{equation}
has also been implicitly or explicitly enforced on occasion \cite{Heinosaari_2018, Bandyopadhyay_2014}, while on other occasions it has not \cite{mishra2023optimal}. This additional condition ensures that all outcomes of the POVM $\{T_{a}\}_{a=1}^{N}$ have some probability of occurring. By enforcing both Eq.~(\ref{stateExclusionEquation}) and Eq.~(\ref{addationalStateExclusionPOVM}), conclusive $1$-state exclusion on the set $\{ \rho_{x} \}^{N}_{x=1}$ is defined to be the existence of an $N$ element POVM where each element excludes a different state from $\{ \rho_{x} \}^{N}_{x=1}$ with certainty, as seen in Fig.~\ref{weak_strong_exlcusion_figure} (a). We define this to be {\em strong state exclusion}.

On the other hand, by only enforcing Eq.~(\ref{stateExclusionEquation}), conclusive $1$-state exclusion on $\{ \rho_{x} \}^{N}_{x=1}$ is defined to be the existence of a POVM with $L$ non-zero elements, where $L\leq N$, such that each conclusively exclude a different state from a subset of $\{ \rho_{x} \}^{N}_{x=1}$ of size $L$, as seen in Fig.~\ref{weak_strong_exlcusion_figure} (b). We define this to be {\em weak state exclusion}.

When extended to $k$-state exclusion, strong exclusion means that there exists a POVM that can exclude all possible sub-sets of $\{ \rho_{x} \}^{N}_{x=1}$ of length $k$. Weak exclusion then means that there exists a POVM that can only exclude {\em only some subsets} of $\{ \rho_{x} \}^{N}_{x=1}$ of length $k$. More details on weak and strong state exclusion can be found in Supplementary Material A.

The task of state exclusion is reminiscent of state discrimination, where the player instead tries to output a label $g$ such that $g=x$. It can be seen that conclusive state discrimination, where a player outputs a label $g=x$ with certainty, is a special case of \ben{strong} conclusive $k$-state exclusion where $k=N-1$. Outputting $N-1$ labels of states that were definitely not sent is equal to outputting one label of the state that definitely was sent. It is a well-known result that conclusive state discrimination, and hence conclusive $(N-1)$-state exclusion, is only possible if all states in $\{ \rho_{x} \}^{N}_{x=1}$ are orthogonal~\cite{nielsen_chuang_2010}.  

A closely related task is sub-channel exclusion. Consider a collection of completely-positive {\em trace-non-increasing} linear maps, $\Psi = \{\Psi_{x}\}_{x=1}^N$, such that $\sum_{x=1}^N \Psi_{x}$ is a channel~\cite{wolf2012quantum}. This collection is called an {\em instrument}, and each map $\Psi_{x}$ is called a {\em sub-channel}. In sub-channel exclusion, a player has a reference state $\rho$ that they send to the referee. The referee then measures $\rho$ using the instrument and returns the post-measurement state to the player. The player measures a POVM on the state and outputs a label $g \in \{1, \ldots ,N\}$. They succeed if they output a label of a sub-channel that was not applied. As before, the player can output the label of a sub-channel not applied with certainty, they can output $k$ labels, $\{ g_{i} \}^{k}_{i=1}$, or they can output $k$ labels with certainty.

{\bf\em Necessary condition for $k$-state exclusion.---}
It has previously been shown that all $k$-state exclusion tasks can be recast as $1$-state exclusion tasks by reformulating the set $\{\rho_{x}\}^{N}_{x=1}$ (see Appendix I of Ref.~\cite{Bandyopadhyay_2014}). \bentwo{Conceptually, this means that all $k$-state exclusion tasks are mathematically equivalent to a $1$-state exclusion task --- solving one task means you have solved the other. Consequently, all state exclusion tasks can be studied within the $1$-state exclusion framework through this ``\textit{equivalent}'' task. }This has led to a consensus that only the task of $1$-state exclusion needed to be studied, and hence, all feasibility conditions in the literature for both weak and strong state exclusion tasks have been for conclusive $1$-state exclusion~\cite{PhysRevA.66.062111, Bandyopadhyay_2014, Heinosaari_2018, mishra2023optimal, Knee_2017, johnston2023tight}. However, when using the reformulation method for accessing $k$-state exclusion tasks, the size of the reformulated sets can get very large for particular values of $N$ and $k$, \ben{potentially} making the \bentwo{equivalent task} computationally \ben{infeasible} to access. In addition, scenarios may exist where one wants to consider the original task rather than the \bentwo{equivalent task}; this may happen, for instance, if the set of states upon which exclusion is being performed holds some operational significance. \ben{By reformulating the set into the \bentwo{equivalent task}, the states in the set will differ, potentially significantly, from the physical states considered in the original task.}
Hence, a condition for $k$-state exclusion that is dependent only on the original set $\{\rho_{x} \}^{N}_{x=1}$ is of \CYtwo{value---it} \ben{by-passes the need to consider the \bentwo{equivalent task}}. Here, a necessary condition of this form is presented as our first main result. It allows for a feasibility test of conclusive $k$-state exclusion where the number of conditions to be checked is always linear in $N$.

\begin{lemma} \label{lemma1}
    A referee has a set of $N$, $d$-dimensional quantum states, $\{ \rho_{x}\}_{x=1}^{N}$. A necessary condition for the existence of a POVM such that the player can perform conclusive strong or weak $k$-state exclusion is  
    \begin{equation}\label{Eq:Lemma1}
        \sum_{x=1}^{N} \Pi_{x} \leq (N-k) \mathbb{I},
    \end{equation}
    where $\Pi_{x}$ is the projector onto the support of $\rho_{x}$ for all $x$.
\end{lemma}
See Appendix I for the proof. Note, given that every sub-channel exclusion task induces an effective state exclusion task, Lemma~\ref{lemma1} can be applied to both state and sub-channel exclusion tasks. The proof of Lemma~\ref{lemma1} does not enforce Eq.~\eqref{addationalStateExclusionPOVM} as a condition, and hence is a provably necessary condition for strong state exclusion~\footnote{Consider the set of states with projectors onto their supports of $\{ \ketbra{00}, \ketbra{01}, \ketbra{00} + \ketbra{10}, \ketbra{01} + \ketbra{11}\}$. These satisfy Lemma ~\ref{lemma1} for $k=2$, but by considering the \bentwo{equivalent task} tasks it can be seen that strong state exclusion is never possible as there is a full-rank state in the reformulated set.}. It is left open as to whether Lemma~\ref{lemma1} is a sufficient condition for weak $k$-state exclusion. However, if the inequality is saturated, then Lemma~\ref{lemma1} is sufficient for both weak and strong $1$-state exclusion. In this case, measuring the POVM $\{(\mathbb{I}-\Pi_{x})/k \}^{N}_{x=1}$ would perform $1$-state exclusion.

As an application of Lemma~\ref{lemma1}, consider $\{ \rho_{x} \}^{N}_{x=1}$ to be a set of $N$ orthogonal states. It follows that $\sum^{N}_{x=1} \Pi_{x} \leq \mathbb{I}$, and, hence, the largest value of $k$ such that Lemma~\ref{lemma1} is satisfied is $k=N-1$. Lemma~\ref{lemma1} therefore implies the ability to perform conclusive state discrimination on a set of $N$ orthogonal states, as expected. This also shows that there exists a set of states for all values of $N$ and $d$ for which \CY{Eq.~\eqref{Eq:Lemma1} is saturated; namely,} Lemma~\ref{lemma1} is tight. In addition, if considering rank $r$ states of dimension $d$, one can always find a weak exclusion task for which 
\CY{Eq.~\eqref{Eq:Lemma1} is saturated}
if $N = { d \choose r }$. Firstly, let $\{ \ket{i} \}^{d-1}_{i=0}$ be a basis in the $d$-dimensional space. 
One can then consider $\{\rho_x\}_{x=1}^N$ such that each $\Pi_x$ is a projector onto the basis elements contained in each subset of$\{ \ket{i} \}^{d-1}_{i=0}$ of length $r$, of which there are $N$ of them. By measuring the POVM $\{ \dyad{i} \}_{i=0}^{d-1}$, conclusive weak $k$-state exclusion can be performed with $k = { d-1 \choose r }$. This is as predicted by Lemma~\ref{lemma1}, as $\sum^{N}_{x=1}\Pi_x = {d-1 \choose r-1} \mathbb{I}$, with ${d \choose r } - {d-1 \choose r } = {d-1 \choose r-1}$. 
Finally, Lemma~\ref{lemma1} also leads to the following corollary on the maximum value of $k$.  
\ben{
\begin{corollary} \label{lowerBoundOnKcorollary}
    When performing conclusive state exclusion on $\{ \rho_{x}\}^{N}_{x=1}$, \bentwo{an upper bound on the number, $k$, of states that can be excluded} is given by 
    \begin{align}
        k \leq N -  \left(2^{D_{\rm max}(\omega \| \mathbb{I}/d)}\alpha\right)/d     \leq  N(d-1)/d ,
    \end{align}
    where $\alpha \coloneqq {\rm tr} \big[ \sum_{x=1}^{N} \Pi_{x} \big]$, $\omega \coloneqq \sum_{x=1}^{N} \Pi_{x}/\alpha$, and $ D_{\textrm{max}} (\psi \| \sigma) \coloneqq \log_2 \min \{\lambda\geq1 : \psi \leq \lambda \sigma \}$ is the max relative entropy~\cite{dattaMaxRealtiveEntropy}.
\end{corollary}
}
See Appendix II for the proof. Corollary~\ref{lowerBoundOnKcorollary} sets a fundamental limit on the number of states that can be excluded, and, interestingly, gives the max-relative entropy a novel operational meaning in terms of state exclusion tasks. \ben{Specifically, given a set of states, the max-relative entropy establishes an upper-bound on the number of states that can be conclusively excluded.} Below, as another main result, the Choi-rank is given a novel operational interpretation --- it sets a universal upper bound on the number of states that be excluded in a communication task.

{\bf\em Operational interpretation of Choi-rank.---}
An operational interpretation of the Choi-rank of a channel $\mathcal{N}$ is now presented through an {\em entanglement-assisted sub-channel exclusion task}. The task is defined as a communication task between two spatially separated parties, Alice (A) and Bob (B). Alice aims to use a pre-shared entangled state to increase the amount of classical information she can send to Bob through a single use of channel $\mathcal{N}$, as in super-dense coding~\cite{Bennett92}. It is assumed that the message Alice is sending is of the utmost importance, meaning Bob chooses to only output information about the encoded message that he is certain of.

\begin{task} \label{task1}
Alice and Bob share a maximally entangled state of local dimension $d$. Alice encodes $x$, one of $N$ bit-strings, that she wants to send to Bob by applying one of the unitary channels from $\{U_{x}\}^{N}_{x=1}$ to her half of the maximally entangled state. She then sends her half of the maximally entangled state to Bob via the channel $\mathcal{N}$. Bob performs a joint measurement and aims to output a set of $k$ bit-strings that he is certain Alice did not encode. 
\end{task}

\ben{
Task~\ref{task1} \CY{aims to} operationally quantify the information about Alice's encoded message \CY{(communicated via $\mathcal{N}$)} that Bob \CY{can} obtain with certainty. \CY{Its} significance can be understood by first considering \CY{the case with $k = N-1$ and $N = d^2$}. If Bob \CY{achieves this, he can} exclude all bit-strings but one, effectively outputting a single bit-string, $l$, such that $l=x$ with certainty. Hence, Bob \CY{knows} with certainty what message Alice encoded, with \CY{Task~\ref{task1}} becoming equivalent to (conclusive) \CY{super-dense coding}. However, as \CY{stated} above, Bob can only achieve $k=N-1$ if \CY{the states} after encoding and sending are \CY{orthogonal}. This \CY{is, in general, not true when $\mathcal{N}$ introduces} noise (see, e.g., Ref.~\cite{HsiehPRXQ2021}). Bob can therefore instead attempt to say something with certainty about Alice's encoded message by outputting a set of $k$ bit-strings, where $k \leq N-1$, which he is certain \CY{\em does not} contain $x$. 
} 

We will focus on Bob's ability to maximise the value of $k$, measuring his success in Task~\ref{task1} by the maximum number of bit-strings that is it possible for him to exclude. The larger the value of $k$, the more Bob knows about which bit-string Alice encoded. This culminates in Bob performing conclusive state discrimination if $k=N-1$ and hence knowing Alice's encoded bit-string with certainty. The following result upper-bounds $k$ via the Choi-rank of $\mathcal{N}$ and holds for all possible unitary-encoding and decoding (POVMs) strategies: 

\begin{theorem} \label{result1}
The maximum number of bit-strings, $k$, that Bob can exclude in Task~\ref{task1} is
    \begin{equation}
        k \leq  N(d^{2}-r^{\mathcal{N}}_{c})/d^{2}, \label{Theorem 1}
    \end{equation}
\end{theorem}
See Appendix III for the proof.
\CY{Note that Eq.~\eqref{Theorem 1} is tight, as the equality can be \bentwo{saturated}. To see this,}
applying Result~\ref{result1}, it can immediately be seen that if $\mathcal{N}$ is a {\em depolarising channel}, $\depola{p}(\rho) \coloneqq p \rho + (1-p)\textrm{tr}[\rho]~\mathbb{I}/d $, then $k=0$. This is because the Choi-states of depolarising channels are full-rank, $\choiRank = d^{2}$, for all $p$. Hence, Bob can say nothing with certainty about the message encoded by Alice when $\mathcal{N}$ is a depolarising channel. Consider instead that Alice is trying to perform super-dense coding, encoding one of $d^{2}$ bit-strings into a maximally entangled state with local dimension $d$ using the {\em Heisenberg-Weyl operators}~\cite{watrous2018theory}. She then sends her half of the state to Bob via a {\em dephasing channel}, \mbox{$\dephase{p}(\rho) \coloneqq p\rho + (1-p)\sum_{n=0}^{d-1} \dyad{n} \rho \dyad{n}$}, which has $\choiRank=d$. Result~\ref{result1} then implies that $k$-exclusion is possible for $k \leq d^2-d$. If Bob measures the POVM that projects into the Bell basis, \CY{Eq.~\eqref{Theorem 1} is saturated in this instance, since Bob can} perform conclusive weak $(d^2-d)$-exclusion. \CY{Hence, Result~\ref{result1} is tight.} See Supplementary Material B for details.  

\bentwo{Result.\ref{result1} is derived under minimal assumptions about dynamics --- namely, that they are modelled by quantum channels --- making it broadly applicable. Moreover, it is fundamental because, whilst tight, one can never improve on it. This fundamental bound on an operational task (Task \ref{task1}) can then be used to provide the Choi-rank with an operational interpretation --- it is the quantity that limits the number of bit-strings that could ever be excluded in Task \ref{task1}.}

\ben{A natural extension to \CY{Result~\ref{result1}} is to ask if it also holds for encodings via general quantum channels in Task~\ref{task1}, rather then just unitary channels. By noting that unital channels are rank-non-decreasing (see Supplementary Material C) and that the transpose of a unital channel is still a unital channel~\cite{stratton2024}, Result~\ref{result1} can initially be expanded to include all encodings via unital channels. We then conjecture that Result~\ref{result1} generalises to include encodings via general (non-untial) channels, leaving the proof for future work.} 


{\bf\em Discussions.---}
We give the Choi-rank a novel operational interpretation as the fundamental limit on \ben{an} entanglement-assisted exclusion tasks. To drive this result, a necessary condition for conclusive $k$-state exclusion has been presented, and the notion of weak and strong state exclusion has been introduced. This condition allows the viability of conclusive $k$-state exclusion to be assessed without the need to first reformulate the set and apply the conditions for $1$-state exclusion. Although, by considering $k=1$, this also adds to the conditions for conclusive $1$-state exclusion already present in the literature~\cite{PhysRevA.66.062111, Bandyopadhyay_2014, Heinosaari_2018, mishra2023optimal, Knee_2017, johnston2023tight}. Whilst it is known that this condition is not sufficient for strong-state exclusion, it would be interesting to know if it is sufficient for weak-state exclusion.

There are several directions in which Result~\ref{result1} could be generalised. Firstly, whilst Result~\ref{result1} holds for all possible unital-encoding and (general) decoding strategies, it is unknown if it holds for all initial states shared between Alice and Bob. It follows from the definition of the Choi-state that Result~\ref{result1} holds for any full-Schmidt-rank state shared between Alice and Bob. And, intuitively, one would imagine that using a less entangled initial state could only reduce one's ability to succeed at the task. This intuition arises from the knowledge that entanglement is a resource for super-dense coding, which is a special case of Task~\ref{task1}. Understanding this would enable us to determine the underlying resources of Task~\ref{task1}. \ben{Secondly, as stated in the above conjecture, Result~\ref{result1} could be generalising to include general (non-unital) encoding strategies.}

Understanding how these extensions affect one's ability to succeed in the task will help assess the boundaries of the limitations imposed by a channel's Choi-rank. Moreover, it will allow the significance of this task in quantifying resources to be assessed~\cite{PhysRevLett.125.110401, Ye2021, uola2020all,hsieh2023quantum,Ducuara2022PRXQ,Ducuara2023,Ducuara2023PRL}.


\begin{acknowledgments}
\CY{\bf\em Acknowledgments.---}B.S.~acknowledges support from UK EPSRC (EP/SO23607/1). P.S.~and C.-Y.H.~acknowledge support from a Royal Society URF (NFQI). C.-Y.H.~also acknowledges support from the ERC Advanced Grant (FLQuant). P.S.~is a CIFAR Azrieli Global Scholar in the Quantum Information Science Programme. 
\end{acknowledgments}

\appendix

\CY{\bf\em Appendix I: Proof of Lemma~\ref{lemma1}.---}
The following lemma is first proved. 
\begin{lemma} \label{lemma2}
    If $\sigma \geq 0$ is some state and $0 \leq Q \leq \mathbb{I}$ some operator such that 
$\textrm{tr} \big[ \sigma Q \big] = 0$,
    then 
        $\Pi_{\sigma} \leq \mathbb{I} - Q,$
    where $\Pi_{\sigma}$ is the projector onto the support of $\sigma$. 
\end{lemma}
\begin{proof}
    Firstly, note that $\sigma \geq \mu_{\rm min}
    (\sigma)\Pi_{\sigma}$, where \mbox{$\mu_{\rm min}(\sigma)>0$} is the minimal positive eigenvalue of $\sigma$. Therefore, 
\mbox{$0 = \textrm{tr} \big[ \sigma Q \big]\geq 
\mu_{\rm min}(\sigma) \textrm{tr} \big[ \Pi_{\sigma} Q  \Pi_{\sigma} \big].$}
   Given $ \mu_{\rm min}(\sigma) > 0$ and $ \Pi_{\sigma} Q  \Pi_{\sigma} \geq 0$, it can be seen that 
$
\Pi_{\sigma} Q  \Pi_{\sigma} = 0. 
$
   Hence, 
    \mbox{$\Pi_{\sigma} \leq \Pi_{\textrm{ker}(Q)}\leq \mathbb{I} -  \Pi_{\textrm{supp}(Q)}$,}
   where $\Pi_{\textrm{ker}(Q)}$ and $\Pi_{\textrm{supp}(Q)}$ are the projectors onto the kernel and support of Q, respectively. Finally, given that $Q \leq \Pi_{\textrm{supp}(Q)}$, we have that
   \mbox{$
    Q \leq \Pi_{\textrm{supp}(Q)} \leq \mathbb{I} - \Pi_{\sigma},
    $}
   completing the proof. 
\end{proof}
The proof of Lemma~\ref{lemma1} is now given, employing Lemma~\ref{lemma2}. 

\begin{proof}
A referee has an set of states  $\{\rho_{x}\}_{x=1}^{N}$. Let $\mathcal{Y}_{(N,k)}$ be the set of all subsets of length $k$ of the set $\{1, \ldots ,N\}$~\cite{Bandyopadhyay_2014}. 
During each round of the task, the referee randomly generates a label $x$ and sends the state $\rho_x$ to the player. The player applies a POVM on the state $\rho_x$ and aims to output a set of $k$ labels $Y \in \mathcal{Y}_{(N,k)}$ such that $x \notin Y$.
Such a measurement will be a POVM with ${N \choose k}$ elements, denoted by ${\bf S} \coloneqq \{S_Y\}_{Y~\in~\mathcal{Y}_{(N,k)}}$. The player is able to perform conclusive $k$-state exclusion if for all $Y~\in~\mathcal{Y}_{(N,k)}$ there exists an $S_{Y}$ such that  
$
\textrm{tr}\left[S_{Y}\rho_{y}\right] = 0$~$\forall ~ y \in Y.
$
If the player gets the measurement outcome associated to $S_{Y}$, they can output the set $Y$ knowing with certainty the referee could not have sent any of the states in the set $\{\rho_{y} \}_{y \in Y}$.
By defining the operator
    $R_{Y} \coloneqq \sum_{y \in Y} \rho_{y},$
the conclusive $k$-state exclusion task can then be succinctly expressed as 
\begin{equation}
    \textrm{tr} \left[ S_{Y} R_{Y}\right] = 0 \quad\forall~Y \in \mathcal{Y}_{(N,k)}. \label{k state exlusion equation}
\end{equation}
Letting $L\coloneqq{N \choose k}$, one can order ${\bf S}$'s elements and write \mbox{${\bf S}\coloneqq\{S_l\}_{l=1}^L$}. 
Similarly, we also order the operators $R_Y$'s by the same label and write $\{R_l\}_{l=1}^L$.
Now, it can be seen that a given $x\in\{1,...~,N\}$ will appear in \mbox{${N-1 \choose k-1}=Lk/N$} many subsets in $\mathcal{Y}_{(N, k)}$. This means that, for each state $\rho_x$, there are $Lk/N$ many operators $R_l$'s that contain it.
For each $x\in\{1,...,N\}$, let $X_x$ denote the set of all labels $l$ corresponding to these $R_l$'s.
Each $X_x$ thus contain $Lk/N$ many labels, and we have ${\rm tr}(S_l\rho_x)=0$ $\forall\,l\in X_x$.
Equation~\eqref{k state exlusion equation} thus implies
\begin{equation}
    \textrm{tr} \Bigg[ \rho_{x} \Bigg( \sum_{l \in X_x} S_{l} \Bigg) \Bigg] = 0 \quad \forall~x \in \{1,...~,N\}.
\end{equation}
Using Lemma~\ref{lemma2}, the above $N$ equations implies
\begin{equation}
     \Pi_{x} \leq \mathbb{I} -  \sum_{l \in X_x} S_{l} \quad \forall~x \in \{1,...~,N\}, \label{projector individual equations}
\end{equation}
where $\Pi_{x}$ is the projector onto the support of the state $\rho_{x}$. Summing over the $N$ individual conditions in Eq.~\eqref{projector individual equations} gives 
\begin{equation}
    \sum_{x=1}^{N} \Pi_{x} \leq N \mathbb{I} - \sum_{x=1}^N\sum_{l \in X_{x}} S_{l}. \label{projector summed equations}
\end{equation}
By Eq.~\eqref{k state exlusion equation}, for each $l$, there are exactly $k$ many possible labels $x$'s such that \mbox{$l\in X_x$}.
Hence, 
each POVM element $S_l$ appears exactly $k$ times, meaning that \mbox{$\sum_{x=1}^N\sum_{l \in X_{x}} S_{l} = \sum_{l=1}^L kS_l = k \mathbb{I}$} and thus \mbox{$\sum_{x=1}^{N} \Pi_{x} \leq (N-k) \mathbb{I}$}, as desired.
\end{proof}

\CY{\bf\em Appendix II: Proof of Corollary~\ref{lowerBoundOnKcorollary}.---}
By comparison to Lemma~\ref{lemma1}, it can be seen that $\lambda$ used in the definition of $D_{\rm max}$ can be related to $N-k$ when considering the largest possible $k$. Rearranging and noting that $k$ must be an integer gives the first inequality. By taking the trace of both sides of Lemma~\ref{lemma1}, the second inequality can be shown. The trace of the left-hand side is lower bounded by $N$, which is achieved when all $\rho_x$'s are rank-one projectors, i.e., all of them are pure states. Once again, the floor is taken to ensure $k$ is an integer. In this best-case scenario where all $\rho_x$'s are rank-one projectors, we have $\alpha = N$. It can then be seen that the second inequality always upper bounds the first given that $0 \leq D_{\textrm{max}} (\psi \| \sigma) ~\forall~\psi, \sigma$.   
\CY{\hfill$\square$}

\CY{\bf\em Appendix III: Proof of Result~\ref{result1}.---}
Alice and Bob share a maximally entangled state $\ket{\Phi^+}_{\rm AB} =\sum_{i=0}^{d-1} \ket{ii}_{\rm AB}/\sqrt{d}$ with an equal local dimension $d$ (here, ${\rm A,B}$ denotes Alice's and Bob's systems). If Alice encodes the bit-string $x$ (via unitary $U_{x,{\rm A}}$ in ${\rm A}$) and then sends her half of the state to Bob via the channel $\mathcal{N}_{\rm A}$, Bob has the state
\begin{align}
\begin{split}
    \rho^{x|\mathcal{N}}_{\rm AB} &\coloneqq (\mathcal{N}_{\rm A} \otimes \mathcal{I}_{\rm B})\circ(U_{x,{\rm A}} \otimes \mathcal{I}_{\rm B}) \big( \maxKB_{\rm AB} \big)\\
    &=(\mathcal{I}_{\rm A} \otimes U_{x,{\rm B}}^{t}) \left(\choi{N}_{\rm AB}\right).
\end{split}
\end{align}
Note that, after the channel $\mathcal{N}_{\rm A}$, Bob has the whole bipartite state.
From Bob's point of view, he, therefore, has a state from the set $\{\rho^{x|\mathcal{N}}_{\rm AB}\}_{x=1}^N$.
Note that all elements of this set have the same rank --- the Choi-rank, $r_{c}^{\mathcal{N}}$, of the channel $\mathcal{N}_{\rm A}$. This is due to the rank of states being invariant under unitary channels. 
Bob now aims to perform conclusive state exclusion on this set and hence Lemma~\ref{lemma1} can be applied. Given all states in the set are of rank $r_{c}^{\mathcal{N}}$, all projectors onto the support of those states are of rank $r_{c}^{\mathcal{N}}$. Taking the trace of both sides of Eq.~\eqref{Eq:Lemma1} in Lemma~\ref{lemma1} therefore gives 
    $
    Nr_{c}^{\mathcal{N}} \leq  (N-k)d^{2}. 
    $
Rearranging and noting that $k$ must be an integer completes the proof. 
\CY{\hfill$\square$}

\bibliographystyle{apsrev4-1}
\bibliography{mainTextBib}

\onecolumngrid










\title{Supplementary Material: Operational Interpretation of the Choi Rank Through Exclusion Tasks}

\author{Benjamin Stratton}
\email{ben.stratton@bristol.ac.uk}
\affiliation{Quantum Engineering Centre for Doctoral Training, H. H. Wills Physics Laboratory and Department of Electrical \& Electronic Engineering, University of Bristol, BS8 1FD, UK}
\affiliation{H.H. Wills Physics Laboratory, University of Bristol, Tyndall Avenue, Bristol, BS8 1TL, UK}

\author{Chung-Yun Hsieh}
\affiliation{H.H. Wills Physics Laboratory, University of Bristol,
Tyndall Avenue, Bristol, BS8 1TL, UK}

\author{Paul Skrzypczyk}
\affiliation{H.H. Wills Physics Laboratory, University of Bristol,
Tyndall Avenue, Bristol, BS8 1TL, UK}
\affiliation{CIFAR Azrieli Global Scholars Program, CIFAR, Toronto Canada}

\date{\today}

\maketitle

\onecolumngrid


\section{Supplementary Material A: Weak and Strong Exclusion Tasks \label{Supplementary MaterialA}}

Within the literature, elements of the definition of state exclusion differ. Here, we present a unifying framework for the different definitions through the notion of weak and strong state exclusion. The complete definitions are restated here for clarity. 
\\
\\
\textbf{Strong State Exclusion}: Given a set of states $\{ \rho_{x} \}^{N}_{x=1}$, strong conclusive $1$-state exclusion is possible if there exists a POVM $T=\{T_{a}\}_{a=1}^{N}$ such that 
    \begin{equation}
         \textrm{tr} \big[ T_{x} \rho_{x} \big] = 0 \quad\forall~x\in\{1, \ldots, N\} \hspace{0.2cm} \textrm{and} \hspace{0.2cm} \sum_{x=1}^{N} \textrm{tr} \big[ T_{a} \rho_{x} ] \neq 0\quad\forall~a\in\{1,\ldots,N\}.
    \end{equation}
\\
\\
\textbf{Weak State Exclusion}: Given a set of states \CYtwo{$\{ \rho_{x} \}^{N}_{x=1}$}, weak conclusive $1$-state exclusion is possible if there exists a POVM $T=\{T_{a}\}_{a=1}^{N}$ such that 
    \begin{equation}
         \textrm{tr} \big[ T_{x} \rho_{x} \big] = 0 \quad\forall~x\in\{1, \ldots, N\}.
    \end{equation}
\\
It is clearly the case from the above definitions that weak state exclusion is a requisite for strong state exclusion --- justifying their respective names. Moreover, if one is able to perform strong state exclusion, they can trivially convert this into weak state exclusion via classical post-processing of the measurement outcomes. It can also be seen that if any of the states in $\{ \rho_{x} \}^{N}_{x=1}$ are full-rank, then strong state exclusion is never possible. This is due to $\textrm{tr}[T_{g}\rho_{x}] = 0$ if and only if $T_{g} = 0$ when $\rho_{x}$ is full-rank.

The above definition of strong conclusive $1$-state exclusion on $N$ states means it is defined as the existence of an $N$ element POVM where each element excludes a different state from $\{ \rho_{x} \}^{N}_{x=1}$ with certainty. It can be the case that some POVM elements exclude multiple states, but each element must exclude at least one different state. The above definition of weak conclusive $1$-state exclusion given above is then defined to be the existence of a POVM with $L$ non-zero elements (where $L\leq N$) that each conclusively exclude a different state from a subset of $\{ \rho_{x} \}^{N}_{x=1}$ of size $L$. Using this terminology, one could define the ability to perform weak state exclusion on $\{ \rho_{x} \}^{N}_{x=1}$ as the ability to perform strong state exclusion on some subset of $\{ \rho_{x} \}^{N}_{x=1}$. 

When considering $k$-state exclusion, strong state exclusion means there exists a POVM that can exclude all possible subsets of $\{ \rho_{x} \}^{N}_{x=1}$ of length $k$, with one subset being excluded with certainty with each measurement. Weak $k$-state exclusion then means that, whilst one subset is still excluded with each measurement, not all subsets of $\{ \rho_{x} \}^{N}_{x=1}$ of length $k$ are excluded. Note, this is equivalent to considering strong and weak exclusion on the equivalent $1$-state exclusion task. 

Previously, our proposed definition of weak state exclusion has been used as the general definition of state exclusion~\cite{mishra2023optimal}. However, this definition has attracted (indirect) criticism for trivialising the problem of state exclusion~\cite{Heinosaari_2018}, as if a player can perform conclusive $1$-state exclusion on any two states $\{\rho_1, \rho_2\}\subseteq\{ \rho_{x} \}^{N}_{x=1}$ using the two-element POVM $\{M_{1}, M_{2}\}$, then by definition $1$-state exclusion could trivially be performed on the whole set $\{ \rho_{x} \}^{N}_{x=1}$ by considering the $N$ element POVM
\begin{equation}
    \{T_{1} = M_{1}, T_{2} = M_{2}, T_{3} = 0, ~\ldots~, T_{N} = 0\}.
\end{equation}
Each measurement outcome would exclude one state, but some states would never be excluded. Whilst such an example is indeed trivial, there exist plenty of intermediate scenarios between this and strong state exclusion that could prove useful in operationally motivated tasks. For example, consider a task where player A is trying to communicate to player B which of $N$ wires can be cut to diffuse a bomb~\cite{PhysRevLett.125.110401}. Player A can aim to send exclusionary information to player B that says ``do not cut wire $\alpha$ or wire $\beta$." They do this by encoding exclusionary information into some quantum state and sending it to player B. Player B then measures a POVM that performs conclusive $2$-state exclusion, allowing them to say with certainty that they do not have states associated to the label $\alpha$ or $\beta$ and hence those are the wires not to cut. To succeed, player B must output a list of wires not to cut; he does not need to be able to exclude all possible subsets of wires of length $2$. Hence, weak $k$-state exclusion would still prove useful in this task at preventing accidental detonation.

\section{Supplementary Material B: Dephasing Channel Example \label{SupplementaryMaterialB}}

Consider that \CYtwo{Alice (A) and Bob (B) share a $d$-dimensional maximally entangled state
$
    \ket{\Phi^{+}}_{\rm AB} =  \sum_{n=0}^{d-1} \ket{nn}_{\rm AB}/\sqrt{d},
$}
and Alice wants to send one of $N=d^{2}$ bit strings to Bob. Alice encodes the bit-string she wants to send to Bob into her half of the maximally entangled state using the so-called Heisenberg-Weyl operators \cite{watrous2018theory} [see Eq.~\eqref{Eq:Weyl} below for their definition]. This is the typical generalisation of super-dense coding to higher dimensions. 
Now, suppose that Alice sends her half of the encoded maximally entangled state to Bob via a {\em dephasing channel}, which is defined by
\begin{equation}
    \dephase{p}(\psi) \coloneqq p\psi + (1-p)\CYtwo{\sum_{n=1}^{d-1} \ketbra{n} \psi \ketbra{n}}, ~ ~ p \in [0,1]. 
\end{equation}
Due to the noise introduced by the dephasing channel, Bob can not say with certainty what bit-string Alice did encode, he can instead try and say which bit-string Alice did not encode. The dephasing channel has a Choi-rank of $\choiRank=d$ (proved below), and hence Result~\ref{result1} implies that {\em $k$-state exclusion is possible for every $k \leq d^2 - d$.} The following lemma shows that Result~\ref{result1} is tight in this scenario, with Bob able to exclude $k = d^2-d$ bit-strings that Alice could have encoded.   

\begin{lemma} 
If Alice encodes one of $d^{2}$ bit-strings into the $d$-dimensional maximally entangled state $\ket{\Phi^{+}}_{\rm AB}$ using the Heisenberg-Weyl operators, and sends her half of the state to Bob using the $d$-dimensional dephasing channel, Bob is able to perform conclusive weak $(d^2-d)$-state exclusion by measuring in the Bell basis (i.e., a basis consisting of maximally entangled states). 
\end{lemma}
\begin{proof}
In a space of dimension $d$, the {\em Heisenberg-Weyl operators} are defined as~\cite{watrous2018theory} 
\begin{equation}\label{Eq:Weyl}
W_{a,b} \coloneqq U^{a}V^{b} = \CYtwo{\sum_{n=0}^{d-1} \Omega^{bn} \ket{n+a}\bra{n},}
\end{equation}
where $a,b\in\{0,1,\ldots,d-1 \}$ are cyclic and,
\begin{equation}
U \coloneqq \CYtwo{\sum_{n=0}^{d-1} \ket{n+1}\bra{n}, ~ ~ V \coloneqq \sum^{d-1}_{n=0} \Omega^{n} \ketbra{n},} ~ ~ \Omega \coloneqq e^{\frac{2 \pi i}{d}}.
\end{equation} 
Using the Heisenberg-Weyl operators, a maximally entangled basis in AB can be generated as 
\begin{equation}
\ket{\Phi_{ab}^{+}}_{\rm AB}\coloneqq(\mathbb{I}_{\rm A} \otimes W_{a,b}) \big( \ket{\Phi^{+}}_{\rm AB} \big) \quad \forall\,a,b\in\{0,1,\ldots,d-1 \}.
\end{equation}
Also, in this notation, we have $\ket{\Phi_{00}^{+}} = \ket{\Phi^{+}}_{\rm AB}$.
In the scenario that we outline above, Bob is performing exclusion on the following set of $d^2$ many bipartite states,
\begin{align}\label{Eq:Bobs states}
    \rho^{(a,b)|\mathcal{N}}_{AB} \coloneqq (\dephase{p} \otimes \mathcal{I}_{\rm B})\circ(\mathcal{W}_{a,b} \otimes \mathcal{I}_{\rm B}) \big( \ketbra{\Phi^{+}}_{\rm AB} \big)= (\mathcal{I}_{\rm A} \otimes \mathcal{W}^{t}_{a,b}) \big( \mathcal{J}^{{\mathcal{D}_p^{\rm ph}}}_{\rm AB} \big) \quad \forall\,a,b\in\{0,1,\ldots,d-1 \},
\end{align}
where $\mathcal{W}_{a,b}[\cdot]\coloneqq W_{a,b}[\cdot]W_{a,b}^\dagger$, $\mathcal{W}^{t}_{a,b}[\cdot]\coloneqq W^{t}_{a,b}[\cdot] W^{t,\dagger}_{a,b}$, and $ \mathcal{J}^{{\mathcal{D}_p^{\rm ph}}}_{\rm AB}$ is the Choi-state of the dephasing channel in ${\rm AB}$ that is given by
\begin{align}
    \mathcal{J}^{ {\mathcal{D}_p^{\rm ph}} }_{\rm AB} \coloneqq (\dephase{p} \otimes \mathcal{I}_{\rm B}) \ketbra{\Phi^{+}}_{\rm AB} = p \ketbra{\Phi^{+}}_{\rm AB} + \frac{(1-p)}{d} \sum_{n=0}^{d-1} \ketbra{nn}_{\rm AB} =  \frac{1}{d} \sum_{n=0}^{d-1} \ketbra{nn}_{\rm AB} + \frac{p}{d} \sum_{n \neq j} \ket{nn} \bra{jj}_{\rm AB}.
\end{align}
By comparison of matrix elements, one can see that this can be rewritten as
\begin{align}
    \CYtwo{\mathcal{J}^{ {\mathcal{D}_p^{\rm ph}} }_{\rm AB} = \alpha \ketbra{\Phi_{00}^{+}}_{\rm AB} + (1-\alpha) \sum_{c=1}^{d-1} \ketbra{\Phi_{0c}^{+}}_{\rm AB},} \label{choiDephasingInBellBasis}
\end{align}
with
$
    \alpha \coloneqq 1 - \frac{1}{d}(d-1)(1-p).
$
From Eq.~\eqref{choiDephasingInBellBasis}, it can be seen that the Choi-rank of the dephasing channel is $d$, that is, $\choiRank=d$ when $\mathcal{N} = \mathcal{D}_p^{\rm ph}$. 
The set of states [i.e., Eq.~\eqref{Eq:Bobs states}] that Bob is performing exclusion on can, therefore, be written as 
\begin{align}
    \rho^{(a,b)|\mathcal{N}}_{AB} &=  (\mathcal{I}_{\rm A} \otimes \Omega^{-ab}\mathcal{W}_{-a,b}) \bigg( \alpha \ketbra{\Phi_{00}^{+}}_{\rm AB} + (1-\alpha) \sum_{c=1}^{d-1} (\mathcal{I}_{\rm A} \otimes \mathcal{W}_{0,c}) \big( \ketbra{\Phi_{00}^{+}}_{\rm AB} \big) \bigg)\quad \forall\,a,b\in\{0,1,\ldots,d-1 \} \\
    &= \Omega^{-ab} \alpha \ketbra{\Phi_{-a,b}^{+}}_{\rm AB} + (1-\alpha) \Omega^{-ab} \sum_{c=1}^{d-1} \ketbra{\Phi_{-a,b+c}^{+}}_{\rm AB} \quad \forall\,a,b\in\{0,1,\ldots,d-1 \}, \label{finalStatesDephasng}
\end{align}
where we have used the identities
        $
        W_{a, b} ^ {t} = \Omega^{-ab}W_{-a, b}
        $
        and 
        \mbox{$
        W_{a,b}W_{n, m} = \Omega^{bn}W_{a + n, b + m} = \Omega^{bn - am}W_{n, m}W_{a,b}
        $}
        for every $a,b,n,m$~\cite{watrous2018theory}.
From Eq.~\eqref{finalStatesDephasng}, it can be seen that any bit-string that Alice encodes using the operators 
\begin{equation}
    \{ W_{-a, b} : b \in\{0,1,\ldots,d-1\} \},
\end{equation}
will output states from the dephasing channel that have identical support. Therefore, if Bob measures the operator that projects into the Bell basis and gets an outcome associated to the POVM element $\ketbra{\Phi^{+}_{rs}}$, he knows Alice must have encoded her bit-string using one of the $d$ operators $\{ W_{r, b} : b ~\in~\{0,1,\ldots,d-1\} \}$, and hence she will have inputted one of the following $d$ states \CYtwo{into} the dephasing channel with certainty:
\begin{equation}
    \{ (\mathcal{I}_{\rm A} \otimes \mathcal{W}_{r,b}) \big( \ketbra{\Phi^{+}}_{\rm AB} \big) : b \in\{0,1,\ldots,d-1\} \}.
\end{equation}
 Bob can, therefore, exclude $k=d^{2}-d$ encoded bit-strings with certainty. As he can only exclude some subsets, this is a weak $(d^2-d)$-state exclusion.  
\end{proof}

\section{Supplementary Material C: Unital Channels are Rank Non-decreasing \label{SupplementaryMaterialC}}

Here, it is shown that the rank of states cannot decrease under unital channels. In what follows, ${\rm rank}(\rho)$ denotes the rank of the state $\rho$, i.e., it is the number of strictly positive eigenvalues that $\rho$ has.
\begin{lemma} \label{unitalChannelLemma}
    If $\mathcal{E}$ is a unital channel, then ${\rm rank}[\mathcal{E}(\rho)] \geq {\rm rank}(\rho)$ $\forall\,\rho$.
\end{lemma}
\begin{proof}
    We first define $\suppVector{\bm{a}} = \{\textrm{number of $i$} : a_{i} > 0\}$ where $a_{i}$ are the components of the vector $\bm{a}$. This notation was introduced in Ref.~\cite{Lostaglio_2018}, where it was noted that given two vectors $\bm{a}, \bm{b}$ representing probability distributions and $p \in [0,1]$,
    \begin{equation}
        \abs{ \textrm{supp}(p\bm{a} + (1-p)\bm{b})} \geq \textrm{max}\{ \suppVector{\bm{a}}, \suppVector{\bm{b}} \}. \label{supportInequaility}
    \end{equation}
    Secondly, we recall the definition of majorisation, where a vector $\bm{x} \in \mathbb{R}^{n}$ majorises a vector $\bm{y} \in \mathbb{R}^{n}$, denoted $\bm{x} \succ \bm{y}$, if
    \begin{equation}
        \sum_{i=1}^{k} x^{\downarrow}_{i} \geq \sum_{i=1}^{k} y^{\downarrow}_{i} ~ ~ \forall ~ ~ k~\in~\{1,\ldots, n\},
    \end{equation}
    where $x_{i}^{\downarrow}$ ($y_{i}^{\downarrow}$) are the components of the vector $\bm{x}$ ($\bm{y}$) ordered in decreasing order, such that $x_{i}^{\downarrow} \geq x_{i+1}^{\downarrow}~\forall~i~\in~\{1, \ldots, n\}$.
    
    Returning to the proof, if there exists a unital channel $\mathcal{E}$ such that $\sigma = \mathcal{E}(\rho)$, then
    \begin{equation}
        \lambda(\rho) \succ \lambda(\sigma),
    \end{equation}
    where $\lambda(\rho)$ and $\lambda(\sigma)$ are vectors of the spectrum of $\rho$ and $\sigma$, respectively~\cite{Gour_2015}. This then implies the existence of a doubly stochastic matrix, $\mathbb{D}$, such that~\cite{bapat_raghavan_1997}
    \begin{align}
        \lambda(\sigma) = \mathbb{D} \lambda(\rho)= \sum_{i} p_{i} \mathbb{P}_{i} \lambda(\rho),
    \end{align}
    where $\sum_i p_{i} = 1, p_i\ge0$ give a probability distribution and $\mathbb{P}_{i}$ are the permutation matrices, as a doubly stochastic matrix is a convex combination of permutation matrices \cite{bapat_raghavan_1997}. Employing Eq.~\eqref{supportInequaility}, it can be seen that 
    \begin{align}
        \suppVector{\lambda(\sigma)} = \suppVector{\sum_{i} p_{i} \mathbb{P}_{i} \lambda(\rho)} \geq \textrm{max}_{i} \{ \suppVector{\mathbb{P}_{i}\lambda(\rho)} \} = \suppVector{\lambda(\rho)}, \label{inequaility in terms of supports}
    \end{align}
    as $\suppVector{\bm{a}}$ is invariant under permutation. To complete the proof, it is noted that 
    $
        \textrm{rank}(\sigma) = \suppVector{\lambda(\sigma)}.
    $
    Hence, Eq.~\eqref{inequaility in terms of supports} gives
    $
        \textrm{rank}(\sigma) \geq \textrm{rank}(\rho).
    $

    Physically, this can be explained by the equal convertibility power of unital channels and noisy operations~\cite{Gour_2015}; unital channels can only output states more or equally as noisy as the input states, and hence, they can only make states more indistinguishable.
\end{proof}


\end{document}